\documentclass[twoside]{article}
\usepackage{natbib}
\usepackage{amsmath}
\usepackage{amsthm}
\usepackage[textwidth=6in,left=1.25in,right=1.25in,asymmetric]{geometry}
\usepackage{times}
\usepackage{bm}
\usepackage{natbib}

\usepackage[plain,noend]{algorithm2e}

\usepackage{amssymb,bbm,graphicx,mathabx,mathrsfs,pdflscape,placeins}
\usepackage{tikz}
\usetikzlibrary{arrows,%
                petri,%
                topaths}%
\usepackage{tkz-berge}
\usepackage{enumitem}
\usepackage{multirow}
\usepackage{fancyhdr}

\theoremstyle{plain}

\newtheorem{theorem}{Theorem}[section]

\newtheorem{proposition}[theorem]{Proposition}
\newtheorem{corollary}[theorem]{Corollary}

\newtheorem{lemma}{Lemma}[section]

\theoremstyle{definition}

\newtheorem{example}[theorem]{Example}

\makeatletter
\renewcommand{\algocf@captiontext}[2]{#1\algocf@typo. \AlCapFnt{}#2} 
\def\@algocf@capt@plain{top}
\renewcommand{\algocf@makecaption}[2]{%
  \addtolength{\hsize}{\algomargin}%
  \sbox\@tempboxa{\algocf@captiontext{#1}{#2}}%
  \ifdim\wd\@tempboxa >\hsize
    \hskip .5\algomargin%
    \parbox[t]{\hsize}{\algocf@captiontext{#1}{#2}}
  \else%
    \global\@minipagefalse%
    \hbox to\hsize{\box\@tempboxa}
  \fi%
  \addtolength{\hsize}{-\algomargin}%
}
\makeatother

\def\T{{ \mathrm{\scriptscriptstyle T} }}

\newcommand{\bs}[1]{\boldsymbol{#1}}
\newcommand{\mc}[1]{\mathcal{#1}}
\def\m{\mathcal}

\newcommand{\ms}[1]{\mathscr{#1}}

\newcommand{\bb}[1]{\mathbb{#1}}

\newcommand{\pr}{\textnormal{pr}}
\newcommand{\Mult}[2]{\textnormal{Multinomial}\left(#1,#2 \right)}
\newcommand{\Diri}[1]{\mc{D}(#1)}

\newcommand{\mb}[1]{\mathbb{#1}}

\newcommand{\tr}{\mbox{tr}}
\def\dmid{ \ \vert \vert \ }
\def\wt{\widetilde}
\def\m{\mathcal}
\def\mb{\mathbb}
\def\mr{\mathrm}

\def\T{{ \mathrm{\scriptscriptstyle T} }}
\newcommand{\vect}[1]{\mbox{vec}(#1)}

\newcommand{\No}[2]{\mbox{Normal}\left( #1, #2 \right)}

\def\T{{ \mathrm{\scriptscriptstyle T} }}
\graphicspath{{./}{Graphics/}}

\begin{document}
\allowdisplaybreaks



\pagestyle{fancy}
\fancyhead[RO,LE]{\small\thepage}
\fancyhead[LO]{J. E. Johndrow and A. Bhattacharya}
\fancyhead[RE]{Optimal credible regions for Bayesian log-linear models}

\title{Optimal Gaussian approximations to the posterior for log-linear models with Diaconis--Ylvisaker priors}
\author{James E. Johndrow\thanks{James Johndrow gratefully acknowledges support from grant ES020619 from the National Institute of Environmental Health Sciences (NIEHS) of the U.S. National Institutes of Health.} \\ Department of Statistical Science, Duke University \\ Durham, North Carolina, USA \\ \texttt{jj@stat.duke.edu} \\ \and Anirban Bhattacharya\thanks{Dr. Bhattacharya acknowledges support for this project from the Office of Naval Research.} \\ Department of Statistics, Texas A\&M University, \\ College Station, Texas, USA \\\texttt{anirbanb@stat.tamu.edu} }


\maketitle

\begin{abstract}
In contingency table analysis, sparse data is frequently encountered for even modest numbers of variables, resulting in non-existence of maximum likelihood estimates. A common solution is to obtain regularized estimates of the parameters of a log-linear model. Bayesian methods provide a coherent approach to regularization, but are often computationally intensive. Conjugate priors ease computational demands, but the conjugate Diaconis--Ylvisaker priors for the parameters of log-linear models do not give rise to closed form credible regions, complicating posterior inference. Here we derive the optimal Gaussian approximation to the posterior for log-linear models with Diaconis--Ylvisaker priors, and provide convergence rate and finite-sample bounds for the Kullback-Leibler divergence between the exact posterior and the optimal Gaussian approximation. We demonstrate empirically in simulations and a real data application that the approximation is highly accurate, even in relatively small samples. The proposed approximation provides a computationally scalable and principled approach to regularized estimation and approximate Bayesian inference for log-linear models.  
\end{abstract}

\begin{keywords}
credible region; conjugate prior; contingency table; Dirichet--Multinomial; Kullback--Leibler divergence; Laplace approximaton.
\end{keywords}

\section{Introduction} \label{sec:intro}
Contingency table analysis routinely relies on log-linear models, which represent the logarithm of cell probabilities as an additive model \citep{agresti2002categorical}. With the standard choice of Multinomial or Poisson likelihood, these are exponential family models, and are routinely fit through maximum likelihood estimation \citep{fienberg2007three}. However, sparsity in the observed cell counts often makes maximum likelihood estimation infeasible (see \cite{haberman1974log} and \cite{bishop2007discrete}) in practical applications. In such cases, regularization is often used to obtain unique parameter estimates \citep{park2007l1,zou2005regularization}. 

A common Bayesian approach to inference in high-dimensional contingency tables is to place a conjugate prior on the parameters of a graphical or hierarchical log-linear model, and an independent prior over the space of all such models (see e.g. \cite{massam2009conjugate}). This leads to a standard model-averaged posterior \citep{hoeting1998bayesian}, where all possible sparse log-linear models in the chosen class are weighted by their posterior evidence. Use of non-conjugate (e.g. Gaussian) priors with computation by Markov chain Monte Carlo \citep{gelfand1990sampling} has also been proposed \citep{dellaportas1999markov}. Although model averaging is generally considered ideal in high dimensional settings, computational algorithms for posterior inference scale exceedingly poorly in $p$. Since the smallest contingency table corresponding to cross-classification of $p$ categorical variables has $2^p$ cells, the corresponding log-linear model has $2^p-1$ free parameters, so the model space grows super-exponentially in $p$. Accordingly, posterior computation is essentially infeasible for $p > 15$, the largest case demonstrated to date in the literature \citep{dobra2010mode} to the best of our knowledge.

Alternatively, one can place a Gaussian prior on the parameters of a saturated log-linear model to induce Tikhonov type regularization, and then perform computation by Markov chain Monte Carlo. This approach is well-suited to situations in which the sample size is not tiny relative to the table dimension, but where zero counts nonetheless exist in some cells. In this case, data augmentation Gibbs samplers such as that proposed by \cite{polson2013bayesian} provide for conditionally conjugate updates. However, this by itself is computationally intensive relative to alternatives such as elastic net \citep{zou2005regularization}, and can suffer from poor mixing. In principle, a more scalable Bayesian approach for producing Tikhonov regularized point estimates would be to utilize the Diaconis--Ylvisaker conjugate prior \citep{diaconis1979conjugate} on the parameters of the log-linear model, which is essentially computation free. The main drawback is that the resulting posterior distribution is difficult to work with, lacking closed form expressions for even marginal credible intervals or fast algorithms for sampling from the posterior. An accurate and more tractable approximation to this posterior is therefore of practical interest.

Approximations to the posterior distribution have a long history in Bayesian statistics, with the Laplace approximation perhaps the most common and simple alternative \citep{tierney1986accurate,shun1995laplace}. More sophisticated approximations, such as those obtained using variational methods \citep{attias1999inferring} may in some cases be more accurate but require computation similar to that for generic EM algorithms. Moreover, there exist no theoretical guarantees of the approximation error in finite samples, and these approximations are known to be inadequate in relatively simple models \citep{wang2004lack,wang2005inadequacy}. 

In this article, we propose a Gaussian approximation to the posterior for log-linear models with Diaconis--Ylvisaker priors. The approximation is shown to be the optimal Gaussian approximation to the posterior in the Kullback--Leibler divergence, and convergence rates to the exact posterior and a finite-sample Kullback--Leibler error bound are provided. The approximation is shown empirically to be accurate even for modest sample sizes; effectively, the empirical results suggest that the approximation is accurate enough to be used in place of the exact posterior within the range of sample sizes for which the posterior is sufficiently concentrated to be statistically useful. We also show how the approximation can be used to perform model selection using the penalized credible region method \citep{bondell2012consistent}. In a real data application, the method performs favorably in model selection for graphical log-linear models compared to methods requiring vastly greater computational resources.

\section{Background} \label{sec:background}
We first provide a brief review of exponential families. We then describe the family of conjugate priors for the natural parameter of an exponential family, referred to as Diaconis--Ylvisaker priors. We then provide more detailed background on log-linear models for Multinomial likelihoods and the associated Diaconis--Ylvisaker prior. 

\subsection{Exponential families}
Following \cite{diaconis1979conjugate}, let $\mu$ be a $\sigma$-finite measure defined on $(\bb R^p,\mc B)$, where $\mc B$ denotes all Borel sets on $\bb R^p$. Let $\mbox{supp}(\mu) = \{ y \in \bb R^p : d \mu(y) > 0\}$ be the support of $\mu$, and define $\ms Y$ as the interior of the convex hull of $\mbox{supp}(\mu)$. For $\theta \in \bb R^p$, define $M(\theta) = \log \int_{\ms Y} e^{\theta^T y} d\mu(y)$, and let $\Theta  = \{\theta \in \bb R^p : M(\theta) < \infty\}$, which we assume is an open set. We refer to $\Theta$ as the natural parameter space. The exponential family of probability measures $\{P(\cdot;\theta)\}$ indexed by a parameter $\theta \in \Theta$ is defined by
\begin{align}
dP(y;\theta) = e^{\theta^T y - M(\theta)} d\mu(y), \qquad \theta \in \Theta.
\end{align}
This family includes many of the probability distributions commonly used as sampling models in likelihood-based statistics. \cite{diaconis1979conjugate} develop the family of conjugate priors for the parameter $\theta$ of regular exponential family likelihoods. These Diaconis--Ylvisaker priors are given by
\begin{align}
d \pi(\theta;n_0,y_0) = e^{n_0 y_0^T\theta - n_0 M(\theta)}, \qquad n_0 \in \bb R, y_0 \in \bb R^d.  
\end{align}
On observing data $y$ consisting of $n$ observations with sufficient statistics $\bar{y}$, the posterior is then also Diaconis--Ylvisaker, with parameters $n_0 + n, y_0 + \bar{y}$, i.e. $d \pi(\theta \mid y) = d \pi(\theta;n_0 + n,y_0 + \bar{y})$. In the sequel we focus on one member of the exponential family, the multinomial. In the natural parametrization, the ultinomial likelihood gives rise to the log-linear model and the closely related multinomial logit model, which we now describe.

\subsection{Log-linear models}

Let $\mc S^d = \{(x_1, \ldots, x_d) \in [0, 1]^d : \sum_{j=1}^d x_j \le 1\}$ denote the $d$-dimensional unit simplex. Consider $N$ independent samples from a categorical variable with $(d+1)$ levels. We denote the levels of the variable by $0, 1, \ldots d$, without loss of generality. Let $y_j$ denote the number of times the $j$th level is observed in the $N$ samples and set $y = (y_0, y_1, \ldots, y_d)^{\T}$; clearly $\sum_{j=0}^d y_j = N$. The joint distribution of $y$ is given by a multinomial distribution, denoted $y \sim \Mult{N}{\pi}$, which is parametrized by $\pi = (\pi_1, \ldots, \pi_d)^{\T} \in \mc S^d$, where $\pi_j$ is the probability of observing the $j$th level for $j = 1, \ldots, d$. 

The log-linear model is a generalized linear model for multinomial likelihoods obtained by choosing the logistic link function, which also results in the natural exponential family parametrization. Define the logistic transformation $\ell : \bb R^d \to \mc S^d$ and its inverse log ratio transformation $\ell^{-1}: \m S^d \to \mb R^d$ as\begin{align}\label{eq:logistic_trans}
\pi_j = \frac{e^{\theta_j}}{1 + \sum_{l=1}^d e^{\theta_l}}, \quad \theta_j = \log(\pi_j/\pi_0), \quad (j = 1, \ldots, d),
\end{align}
where $\pi_0 = 1 - \sum_{j=1}^d \pi_j$, and $\theta_0 = 0$. We shall write $\pi = \ell(\theta)$ and $\theta = \ell^{-1}(\pi) = \log(\pi/\pi_0)$, respectively, to denote the transformations in \eqref{eq:logistic_trans}. Using \eqref{eq:logistic_trans}, the multinomial likelihood in the log-linear parameterization can be expressed as 
\begin{align}\label{eq:mult_lik}
f(y \mid \theta) ~\propto~ \frac{\exp\big(\sum_{j=1}^d y_j \theta_j \big)}{\big(1 + \sum_{l=1}^d e^{\theta_l}\big)^N}.
\end{align}

An important motivating case is when $y = \vect{\mathbf{n}}$, with $\mathbf{n}$ a contingency table arising from cross-classification of $N$ independent observations on $p$ categorical variables $y_1, \ldots, y_p$. Suppose that the $v$th variable $y_v$ has $d_v$ many levels, so that the contingency table has $\prod_{v=1}^p d_v$ many {\em cells}, and $y$ is a $(d+1)$-dimensional vector of counts with $d = \prod_{v=1}^p d_v - 1$. We refer to the parametrization $\theta = \log (\pi/\pi_0)$ in the contingency table setting as the {\em identity} parametrization. Also of particular interest in this setting are reparametrizations of \eqref{eq:logistic_trans} that represent $\log \pi/\pi_0$ as an additive model involving parameters that correspond to interactions among $y_1,\ldots,y_p$. Every identified parametrization of the log-linear model for the multinomial likelihood can be represented by 
\begin{align}
\log(\pi/\pi_0) = X \theta^*, \label{eq:llmX}
\end{align}
where $X$ is a $d$ by $d$ non-singular binary matrix and $\theta^* \in \bb R^d$. In the simulations and application, we make a specific choice for $X$ that corresponds to the \emph{corner parametrization} of the log-linear model \citep{massam2009conjugate}. We illustrate the identity and corner parameterizations through a $2^3$ contingency table in Example \ref{ex:3_way} below. 
Details for the general case can be found in the Appendix. 
\begin{example}\label{ex:3_way}
Consider three binary variables $y_1, y_2, y_3$, with $y_v \in \{0, 1\}$ for $v = 1, 2, 3$, and let
$$
\psi_{i_1 i_2 i_3} = \mbox{pr}(y_1 = i_1, y_2 = i_2, y_3 = i_3), \quad (i_1, i_2, i_3) \in \{0, 1\}^3.
$$
A $2^3$ contingency table $\mathbf{n} = (n_{i_1 i_2 i_3})$ is obtained from the cross-classification of $N$ independent observations on $y_1, y_2, y_3$, with $n_{i_1 i_2 i_3}$ denoting the cell count for the cell $(i_1, i_2, i_3)$. Let $y = \vect{\mathbf{n}} = (n_{000}, \ldots, n_{111})^{\T}$ be the vectorized cell counts with $d = 7$. 
In the {\em identity} parametrization, the vector of log-linear parameters $\theta \in \bb R^7$ is given by 
\begin{align*}
\left( \begin{array}{c} \theta_1 \\ \theta_2 \\ \theta_3 \\ \theta_4 \\ \theta_5 \\ \theta_6 \\ \theta_7 \end{array} \right) = 
\log \left( \begin{array}{c} \pi_1/\pi_0 \\ \pi_2/\pi_0 \\ \pi_3/\pi_0 \\ \pi_4/\pi_0 \\ \pi_5/\pi_0\\ \pi_6/\pi_0 \\ \pi_7/\pi_0 \end{array} \right) = 
\log \left( \begin{array}{c} \psi_{001}/\psi_{000} \\ \psi_{010}/\psi_{000} \\ \psi_{011}/\psi_{000} \\ \psi_{100}/\psi_{000} \\ \psi_{101}/\psi_{000} \\ \psi_{110}/\psi_{000} \\ \psi_{111}/\psi_{000} \end{array} \right).
\end{align*}
On the other hand, in the {\em corner} parametrization, we express
\begin{align*}
\theta = \log \left( \begin{array}{c} \psi_{001}/\psi_{000} \\ \psi_{010}/\psi_{000} \\ \psi_{011}/\psi_{000} \\ \psi_{100}/\psi_{000} \\ \psi_{101}/\psi_{000} \\ \psi_{110}/\psi_{000} \\ \psi_{111}/\psi_{000} \end{array} \right) &= \left( \begin{array}{c} \theta^*_{001} \\ \theta^*_{010} \\ \theta^*_{001} + \theta^*_{010} + \theta^*_{011} \\ \theta^*_{100} \\ \theta^*_{001} + \theta^*_{100} + \theta^*_{101} \\ \theta^*_{010} + \theta^*_{100} + \theta^*_{110} \\ \theta^*_{001} + \theta^*_{010} + \theta^*_{100} + \theta^*_{011} + \theta^*_{101} + \theta^*_{110} + \theta^*_{111} \end{array} \right) \\
 &= \left( \begin{array}{ccccccc} 1 & 0 & 0 & 0 & 0 & 0 & 0 \\ 0 & 1 & 0 & 0 & 0 & 0 & 0 \\ 1 & 1 & 1 & 0 & 0 & 0 & 0 \\ 0 & 0 & 0 & 1 & 0 & 0 & 0 \\ 1 & 0 & 0 & 1 & 1 & 0 & 0 \\ 0 & 1 & 0 & 1 & 0 & 1 & 0 \\ 1 & 1 & 1 & 1 & 1 & 1 & 1 \end{array} \right) \times \left( \begin{array}{c} \theta^*_{001} \\ \theta^*_{010} \\ \theta^*_{011} \\ \theta^*_{100} \\ \theta^*_{101} \\ \theta^*_{110} \\ \theta^*_{111} \end{array} \right) \\
 &= X \theta^*.
\end{align*}
The indexing of the elements of $\theta^*$ by binary indices is for ease of interpretation. Indeed, entries of $\theta^*$ with a single $1$ in the binary index are main effects, those with two $1$'s are two-way interactions and $\theta^*_{111} $ is a three-way interaction term. The matrix $X$ can be easily verified to be non-singular, so that the $\theta$ and $\theta^*$ parametrizations are equivalent, with $d = 7$ free parameters in either case. 
\end{example}

\subsection{Conjugate priors for log-linear models}

We now present the Diaconis--Ylvisaker prior for the multinomial likelihood \eqref{eq:mult_lik} and derive an optimal Gaussian approximation to the corresponding posterior in Kullback--Leibler divergence. Extensions to log-linear models with a non-identity parametrization (i.e., $X \ne \mathrm{I}_d$ in \eqref{eq:llmX}) is straightforward by invariance properties of the Kullback--Leibler divergence and are discussed subsequently. All proofs are deferred to the Appendix.

For the multinomial likelihood \eqref{eq:mult_lik}, the Diaconis--Ylvisaker prior is obtained by applying the inverse logistic transformation $\ell^{-1}$ to a Dirichlet distribution, which not surprisingly is the conjugate prior for $\pi$. Recall that $\pi_0 = 1 - \sum_{j=1}^d \pi_j$. The Dirichlet distribution $\Diri{\alpha}$ on $\m S^d$ with parameter vector $\alpha = (\alpha_0, \alpha_1, \ldots, \alpha_d)^{\T}$ has density
\begin{align}\label{eq:dir_dens}
q(\pi; \alpha) = \frac{\Gamma(\sum_{j=0}^d \alpha_j)}{\prod_{j=0}^d \Gamma(\alpha_j)} \, \prod_{j=0}^d \pi_j^{\alpha_j - 1}, \quad \pi \in \m S^d,
\end{align}
and corresponding probability measure $\mc{Q}(\cdot, \alpha)$ with $\mc Q(A, \alpha) = \int_A q(\pi; \alpha) d \pi$ for Borel subsets $A$ of $\m S^d$. 
\begin{proposition} \label{prop:diridy}
Suppose $\pi \sim \Diri{\alpha}$ and let $\theta = \log(\pi/\pi_0) \in \bb R^d$. Define  $A = \sum_{j=0}^d \alpha_j$. Then $\theta$ has a density on $\bb R^d$ given by 
\begin{align}\label{eq:dyprior}
p(\theta; \alpha) = \frac{\Gamma(\sum_{j=0}^d \alpha_j)}{\prod_{j=0}^d \Gamma(\alpha_j)} \, \frac{\exp( \sum_{j=1}^d \alpha_j \theta_j )}{ (1  + \sum_{l=1}^d e^{\theta_l} )^{A} }.
\end{align} 
\end{proposition}
We write $\theta \sim \mc{LD}(\alpha)$ and use $\mc P(\cdot;\alpha)$ to denote the probability measure associated with the density \eqref{eq:dyprior}, with $\mc P(B;\alpha) = \int_B p(\theta;\alpha) d\theta$ for Borel subsets $B$ of $\mb R^d$. If a non-identity parametrization $\theta = X \theta^*$ as in \eqref{eq:llmX} is employed, then we denote the induced distribution on $\theta^* = X^{-1} \theta$ by $\mc P_X(\cdot; \alpha)$ and the density by $p_X(\theta; \alpha)$.

It is immediate that $\mc{LD}(\alpha)$ is a conjugate family of prior distributions for the likelihood \eqref{eq:mult_lik}, with the posterior $\theta \mid y \sim \mc{LD}(\alpha + y)$. To obtain some preliminary insight into the distribution family $\mc{LD}(\alpha)$, we derive the mean and covariance in Proposition 2 below.  
\begin{proposition}\label{prop:meancov}
Let $\theta \sim \mc{LD}(\beta)$, with $\beta = (\beta_0, \beta_1, \ldots, \beta_d)^{\T}$ and $\beta_j > 0$ for all $j$. Then, 
\begin{align*}
E(\theta_j) &= \psi(\beta_j) - \psi(\beta_0), \quad (j = 1, \ldots, d) \\
\mathrm{cov}(\theta_j, \theta_{j'}) &= \psi'(\beta_j) \delta_{jj'} + \psi'(\beta_0), \quad (j, j' = 1, \ldots, d)
\end{align*}
where $\psi$ and $\psi'$ are the digamma and trigamma functions, respectively, and $\delta_{jj'}=0$ if $j \ne j'$ and $\delta_{jj'} = 1$ otherwise.
\end{proposition}
The proof of Proposition \ref{prop:meancov} is established within the proof of Theorem \ref{thm:KL} in the Appendix. Assume the data $y$ is generated from a $\Mult{N}{\pi^0}$ distribution and let $\theta^0 = \log (\pi^0/\pi^0_0 )$ be the true log-linear parameter, where $\pi_0^0 = 1 - \sum_{j=1}^d \pi_j^0$. If a $\mc{LD}(\alpha)$ prior is placed on $\theta$, one can use Proposition \ref{prop:meancov} to show that the posterior mean $E(\theta \mid y)$ converges almost surely to $\theta^0$ with increasing sample size, and the posterior covariance $\mbox{cov}(\theta \mid y)$ converges to the inverse Fisher information matrix as long as the entries of the prior hyperparameter $\alpha$ are suitably bounded. In fact, a Bernstein--von Mises type result can be established, showing that the posterior distribution approaches a Gaussian distribution, centered at the true parameter value and having covariance the inverse Fisher information matrix, in the total variation metric. We do not pursue such frequentist asymptotic validations further in this paper. Our goal rather is to provide a Gaussian approximation to the posterior distribution that can be used in practice, and provide finite sample bounds to the approximation error. 

\section{Main results} \label{sec:results}

In this section, we provide an optimal Gaussian approximation to a $\mc{LD}(\beta)$ distribution \eqref{eq:dyprior} in the Kullback--Leibler divergence, i.e., we exhibit a vector $\mu^* \in \bb R^d$ and a positive definite matrix $\Sigma^*$ such that the Kullback--Leibler divergence between $\mc{LD}(\beta)$ and $\mc N(\mu^*, \Sigma^*)$ is the minimum among all Gaussian distributions. This result provides a readily available Gaussian approximation to the posterior distribution $\mc{LD}(\beta = \alpha + y)$ of the log-linear parameter $\theta$ in \eqref{eq:mult_lik} with a Diaconis--Ylvisaker prior $\mc{LD}(\alpha)$. We also provide a non-asymptotic error bound for the Kullback--Leibler approximation. Using Pinsker's inequality, the approximation error in the total variation distance can be bounded in finite samples. 

For two probability measures $\nu \ll \nu^*$, we write
\begin{align*}
 D(\nu \dmid \nu^*) = E_{\nu^*} \log d\nu/d\nu^* 
\end{align*}
to denote the Kullback--Leibler divergence between $\nu$ and $\nu^*$. 
\begin{theorem} \label{thm:KL}
Given $\beta_j > 0,  j = 0, 1, \ldots, d$, let $\beta = (\beta_0, \ldots, \beta_d)^{\T}$, and define
\begin{align}\label{eq:min_mus}
\mu_j^* = \psi(\beta_j) - \psi(\beta_0), \quad \sigma_{jj'}^* = \psi'(\beta_j) \delta_{jj'} + \psi'(\beta_0),
\end{align}
where $\psi$ and $\psi'$ denote the digamma and trigamma functions respectively. Define $\mu^* = (\mu_j^*) \in \mb R^d$ and $\Sigma^* = (\sigma_{jj'}^*) \in \mb R^{d \times d}$. Then, 
\begin{align}
D \bigg\{ \mc{LD}(\beta) \dmid \mc{N}(\mu^*, \Sigma^*) \bigg\} = \inf_{\mu, \Sigma} D \bigg\{ \mc{LD}(\beta) \dmid \mc{N}(\mu, \Sigma) \bigg\},
\end{align}
where the infimum is over all $\mu \in \mb R^d$ and all $\Sigma \succ 0 \in \mb R^{d \times d}$. Further, if $\beta_j > 1/2$ for all $j = 0, 1, \ldots, d$, then
\begin{align}
D \bigg\{ \mc{LD}(\beta) \dmid \mc{N}(\mu^*, \Sigma^*) \bigg\} < \frac{1}{2} \sum_{j=0}^d \frac{1}{\beta_j} + \frac{1}{6 B}, \label{eq:klrate}
\end{align}
where $B = \sum_{j=0}^d \beta_j$.
\end{theorem}
The matrix $\Sigma^*$ has a compound-symmetry structure and is therefore positive-definite. From Proposition \ref{prop:meancov}, the parameters of the optimal Gaussian approximation $\mu^*$ and $\Sigma^*$ are indeed the mean and covariance matrix of the $\mc{LD}(\beta)$ distribution. Equation \eqref{eq:klrate} provides an upper-bound to the approximation error. In the posterior, $\beta_j = \alpha_j + y_j$ and $B = \sum_{j=0}^d \alpha_j + N$. The condition $\beta_j \ge 1/2$ is therefore satisfied whenever every category has at least one observation. Since
\begin{align*}
\bb E_{y}[\alpha_j+y_j] = \alpha_j + N \pi_j^0,
\end{align*}
the approximation error is approximately in the order of $\sum_{j=0}^d (\pi_j^0 N)^{-1}$, where as before $\pi_j^0$ denotes the true probability of category $j$. In the best case where all the categories receive approximately equal probability, i.e., $\pi_j^0 \asymp (d+1)^{-1}$, the approximation error is $\mc O(d^2/N)$. However, the convergence rate in $N$ can be slower if some of the $\pi_j^0$s are very small. In other words, the higher the entropy of the data generating distribution, the worse the approximation is, although our simulations suggest that the approximation is practicable even for moderate sample sizes and unbalanced category probabilities. When one considers that the eigenvalues of the covariance matrix enter into the constant in Berry-Ess\'{e}en convergence rates, and that here the covariance of the data is given by $\mbox{diag}(\pi^0) - \pi^0 (\pi^0)^{\T}$, it appears that a similar phenomenon is at work here. 

The main idea behind our proof is to exploit the invariance of the Kullback--Leibler divergence under bijective transformations and transfer the domain of the problem from $\bb R^d$ to $\mc S^d$. Since an $\mc{LD}(\beta)$ distribution is obtained from a Dirichlet $\mc{D}(\beta)$ distribution via the inverse log-ratio transform $\ell^{-1}$, the problem of finding the best Gaussian approximation to $\mc{LD}(\beta)$ is equivalent to finding the best approximation to $\mc{D}(\beta)$ among a class of distributions obtained by applying the logistic transform to Gaussian random variables. If $\theta \sim N(\mu, \Sigma)$, the induced distribution on $\pi = \ell(\theta)$ is called a logistic normal distribution -- denoted $\mc L(\mu, \Sigma)$ -- and has density on $\m S^d$ given by 
\begin{align}\label{eq:ln_dens}
\wt{q}(\pi; \mu, \Sigma) = (2 \pi)^{-d/2} |\Sigma|^{-1/2} \bigg( \prod_{j=0}^d \pi_j \bigg)^{-1} \exp \left[ -\frac{1}{2} \{\log(\pi/\pi_0) - \mu \}^{\T} \Sigma^{-1} \{\log(\pi/\pi_0) - \mu \} \right].
\end{align}
The problem therefore boils down to calculating the Kullback--Leibler divergence between a Dirichlet density $q(\cdot; \beta)$ and a logistic normal density $\wt{q}(\cdot; \mu, \Sigma)$ and optimizing the expression with respect to $\mu$ and $\Sigma$. The details are deferred to the Appendix. 

Once the approximation is derived in the identity parametrization, we appeal to the invariance of the Kullback--Leibler divergence under one-to-one transformations to obtain the corresponding approximation in a non-identity parameterization $\theta = X \theta^*$ as in \eqref{eq:llmX} for any non-singular $X$. The result is stated below. 
\begin{corollary}\label{cor:KL}
 If $\theta \sim \mc{LD}(\beta)$ then
 \begin{align}
  D\left( \mc P_X(\cdot;\beta) \dmid \mc N(\cdot;X \mu^*,X^T \Sigma^* X) \right) = \inf_{\mu, \Sigma} D\left( \mc P_X(\cdot;\beta) \dmid \mc N(\cdot;\mu,\Sigma) \right)
 \end{align}
 for any full-rank $d$ by $d$ matrix $X$. Moreover, the bound on the KL divergence as a function of $\beta$ in (\ref{eq:klrate}) is attained for $ D\left(\mc P_X(\cdot;\beta) \dmid \mc N(\cdot;\mu^*,\Sigma^*)\right)$
\end{corollary}
Thus, the best Gaussian approximation to the posterior (in the Kullback--Leibler sense) under the Diaconis--Ylviaker prior is given by $N(X \mu^*,X' \Sigma^* X)$ for any one-to-one linear transformation $X$. We refer to this as the optimal Normal (oN) approximation. 

\section{Simulations} \label{sec:sim}
We conducted several simulation studies to assess the performance of the approximation in Theorem \ref{thm:KL} and Corollary \ref{cor:KL}. In each study, we simulated 100 realizations from
\begin{align}
\pi \sim \Diri{a,\ldots,a}, \quad y \sim \Mult{N}{\pi}, \label{eq:simmodel}
\end{align}
with the posterior of $\pi$ under a Dirichlet $\Diri{a,\ldots,a}$ prior being $\Diri{y_1+a,\ldots,y_d+a}$. 
We chose the dimension $d$ to be $2^8$, corresponding to a $p=8$-way contingency table for binary variables. To obtain a simulation-based approximation to the posterior for $\theta = \log(\pi/\pi_0)$ under the Diaconis--Ylvisaker prior, we sampled $mc$ many $\pi$ values from the $\Diri{y_1+a,\ldots,y_d+a}$ posterior and then transformed to $\theta = \ell^{-1}(\pi)$ to obtain posterior samples of $\theta$; we refer to this procedure as the Monte Carlo approximation. We also computed a Laplace approximation to the posterior under the Diaconis--Ylvisaker prior, which is given by $\No{\hat{\theta}_{MAP}}{\mc I(\hat{\theta}_{MAP})^{-1}}$, where $\hat{\theta}_{MAP}$ is the {\em maximum a-posteriori} estimate of $\theta$ and $\mc I(\theta)$ is the Fisher information matrix evaluated at $\theta$. The maximum a-posteriori estimate $\hat{\theta}_{MAP}$ was computed by the Newton--Raphson method. 

We compare the accuracy of the proposed Gaussian approximation to the Monte Carlo procedure and the Laplace approximation. In addition to the identity parameterization, i.e., $X = \mr I_d$ in \eqref{eq:llmX}, we also consider the corner parameterization given by $\log(\pi/\pi_0) = X \theta^*$ for an appropriate $X$ matrix; see Appendix for more details. For the Monte Carlo samples, each sample of $\theta$ is transformed to $\theta^*$ via $X^{-1} \theta = \theta^*$. For the normal approximations $\theta \sim \No{\mu}{\Sigma}$, the corresponding approximate posterior is given by $\theta^* \sim \No{X^{-1} \mu}{X^{-1} \Sigma X^{-1}}$.

%
%

We conduct simulations for different values of $N$ (250, 1000, and 10,000) and $a$ ($1$ and $1/d$). We then assess performance in several ways. 
\begin{itemize}[labelindent=*,leftmargin=*]
 \item Proportion of variation unexplained, measured by $\sqrt{\sum_{j=1}^d (\theta - \theta_0)^2}/\mbox{sd}(\theta_0)$, where $\theta_0$ is the true value of $\theta$ (or $\theta^*$, as appropriate).
 \item Coverage of 95 percent posterior credible intervals for $\theta$ or $\theta^*$.
 \item The standardized loss in the Frobenius norm for estimates of $\Sigma$, the posterior covariance, given by $||\widehat{\Sigma}-\Sigma||_F/||\Sigma||_F$, where $||S||_F$ is the Frobenius norm of $S$. Note that the covariance in Theorem \ref{thm:KL} is exactly the posterior covariance, so this measure is computed only for the simulation and Laplace approximations. 
 \item The value of the Kolmogorov-Smirnov statistic for comparing the Monte Carlo empirical measure $\frac{1}{mc} \sum_{t=1}^{mc} \delta_{\theta_t}$ to the normal approximation from Theorem \ref{thm:KL}, $\No{\mu}{\Sigma}$. 
 \item The computation time required to compute each posterior approximation.
\end{itemize}

Table \ref{tab:rmse} shows unexplained variation for the Laplace approximation, the Monte Carlo approximation for $mc=10^3, 10^4, 10^5,$ and $10^6$, and the optimal normal approximation. As expected, the optimal normal approximation outperforms the Laplace approximation. Moreover, it is comparable to the Monte Carlo approximation at every sample size and for all of the values of $mc$ considered. Performance for all approximations is noticeably better in the corner parametrization than the identity parametrization.

\begin{table}[h]
\centering
\caption{$\sqrt{\sum_{j=1}^d (\theta - \theta_0)^2}/\mbox{sd}(\theta_0)$ for different values of $mc$, different sample sizes, and two parametrizations. Results are averaged over 100 replicate simulations for each sample size.} \label{tab:rmse}
\begin{tabular}{lcccccc}
\hline
& Laplace & $mc=10^3$ & $mc=10^4$ & $mc=10^5$ & $mc=10^6$ & oN \\
\hline
identity, N=250 & 1.08 & 0.98 & 0.98 & 0.98 & 0.98 & 0.98 \\
corner, N=250 & 0.85 & 0.81 & 0.81 & 0.81 & 0.81 & 0.81 \\
identity, N=1000 & 0.84 & 0.77 & 0.77 & 0.77 & 0.77 & 0.77 \\
corner, N=1000 & 0.67 & 0.61 & 0.61 & 0.61 & 0.61 & 0.61 \\
identity, N=10,000 & 0.40 & 0.35 & 0.35 & 0.35 & 0.35 & 0.35 \\
corner, N=10,000 & 0.31 & 0.27 & 0.27 & 0.27 & 0.27 & 0.27 \\
\hline
\end{tabular}
\end{table}

Table \ref{tab:cover} shows coverage of approximate 95 percent credible intervals for the Laplace approximation, optimal Normal approximation, and the Monte Carlo approximation. The intervals derived using the Laplace approximation are universally too wide. Nominal coverage for the Monte Carlo approximation is insensitive to the value of $mc$ in the range tested, and is slightly high at the two smaller sample sizes. The optimal normal approximation has the best coverage; in all cases it is between 0.94 and 0.96 and for $N=10,000$ the coverage is 0.95 in both parametrizations. 

\begin{table}[h]
\centering
\caption{coverage of 95\% posterior credible intervals} \label{tab:cover}
\begin{tabular}{lcccccc}
\hline
& Laplace & $mc=10^3$ & $mc=10^4$ & $mc=10^5$ & $mc=10^6$ & oN \\
\hline
identity, N=250 & 0.95 & 0.97 & 0.97 & 0.97 & 0.97 & 0.96 \\
corner, N=250 & 1.00 & 0.96 & 0.96 & 0.96 & 0.96 & 0.96 \\
identity, N=1000 & 0.98 & 0.96 & 0.96 & 0.96 & 0.96 & 0.96 \\
corner, N=1000 & 1.00 & 0.94 & 0.94 & 0.94 & 0.94 & 0.94 \\
identity, N=10,000 & 1.00 & 0.95 & 0.95 & 0.95 & 0.95 & 0.95 \\
corner, N=10,000 & 1.00 & 0.95 & 0.95 & 0.95 & 0.95 & 0.95 \\
\hline
\end{tabular}
\end{table}

Table \ref{tab:frob} shows dependence of $||\widehat{\Sigma}-\Sigma||_F/||\Sigma||_F$ on $mc$ for the two different parametrizations and three sample sizes considered. Note that $\Sigma$ is known exactly since $\Sigma=\Sigma^*$, the posterior covariance under the DY prior. The main point of this table is to demonstrate the relatively large number of Monte Carlo samples required to obtain reasonably small error in estimation of the posterior covariance. Even with $10^5$ samples the relative error is on the 1 percent range. Thus, compound linear hypothesis testing and computation of credible regions is very inefficient using the Monte Carlo method. 

\begin{table}[h]
\centering
\caption{$||\widehat{\Sigma}-\Sigma||_F/||\Sigma||_F$ for different sample sizes and values of $mc$}.  \label{tab:frob}
 \begin{tabular}{lcccc}
\hline
& $mc=10^3$ & $mc=10^4$ & $mc=10^5$ & $mc=10^6$ \\
\hline
identity, N=250 & 0.0982 & 0.0328 & 0.0093 & 0.0032 \\
corner, N=250 & 0.0923 & 0.0290 & 0.0086 & 0.0029 \\
identity, N=1000 & 0.1045 & 0.0330 & 0.0103 & 0.0035 \\
corner, N=1000 & 0.0882 & 0.0277 & 0.0087 & 0.0029 \\
identity, N=10,000 & 0.1231 & 0.0397 & 0.0118 & 0.0040 \\
corner, N=10,000 & 0.0861 & 0.0280 & 0.0084 & 0.0027 \\
\hline
\end{tabular}
\end{table}

Table \ref{tab:time} shows the computation time in seconds for each of the three approximations. The Laplace approximation is fast, requiring about 0.03-0.04 seconds to compute at all sample sizes. The optimal normal approximation is about an order of magnitude faster, with the computation time arising mainly in computing the polygamma functions and matrix multiplications. The Monte Carlo approximation is about four orders of magnitude slower than the optimal Normal approximation. Here, only $mc=10^6$ is considered because of the non-negligible error in the posterior covariance for smaller samples; the algorithm scales linearly in $mc$ so for $mc=10^5$ the required time would be approximately 3 seconds. Only about 100 samples could be obtained in the 0.003 seconds required to compute the optimal normal approximation.

\begin{table}[h]
\centering
\caption{Average time (seconds) to compute each approximation, averaged over 100 replicate simulations for each sample size.} \label{tab:time}
 \begin{tabular}{lccc}
\hline
& Laplace & $mc=10^6$ & oN \\
\hline
N=250 & 0.037 & 32.652 & 0.003 \\
N=1000 & 0.031 & 31.980 & 0.003 \\
N=10,000 & 0.035 & 32.338 & 0.003 \\
\hline
\end{tabular}
\end{table}

Results in the previous tables make clear that the optimal normal approximation is superior to the other approximations considered in terms of point estimation, estimation of 95 percent credible intervals, covariance estimation, and computation time. However, it is possible that differences between the optimal normal approximation and the exact posterior exist in the tails of the distribution. To assess this, we compare the empirical measure of the Monte Carlo approximation using $mc=10^6$ samples to the optimal normal approximation by computing the Kolmogorov-Smirnov (KS) statistic for the marginal distributions of 20 randomly selected entries of $\theta$. The entries considered were re-selected for each of the 100 replicate simulations and for each of the three sample sizes. Shown in Figure \ref{fig:ks} are histograms of these KS statistics in the corner and identity parametrizations. Most are less than 0.02, and none are greater than 0.07. Considering that the KS statistic is a point estimate of the total variation distance between distributions, this indicates that the optimal normal approximation is an excellent approximation to the posterior marginals. Moreover, we cannot rule out the possibility of residual Monte Carlo error in the marginals from the Monte Carlo approximation, which may account for part of the observed discrepancy. 

\begin{figure}[h]
\centering
 \begin{tabular}{cc}
  Kolmogorov-Smirnov -- identity & Kolmogorov-Smirnov -- corner \\
  \includegraphics[width=0.4\textwidth]{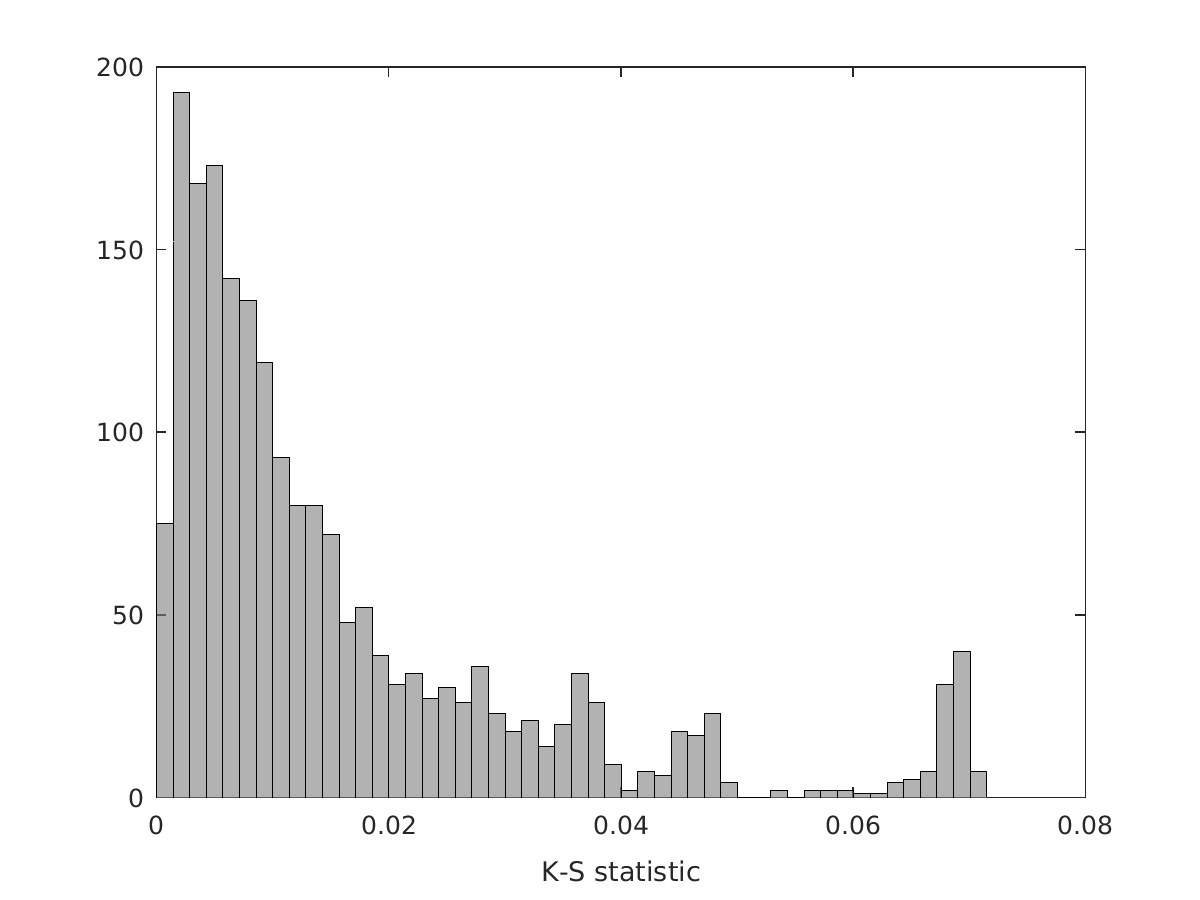} & \includegraphics[width=0.4\textwidth]{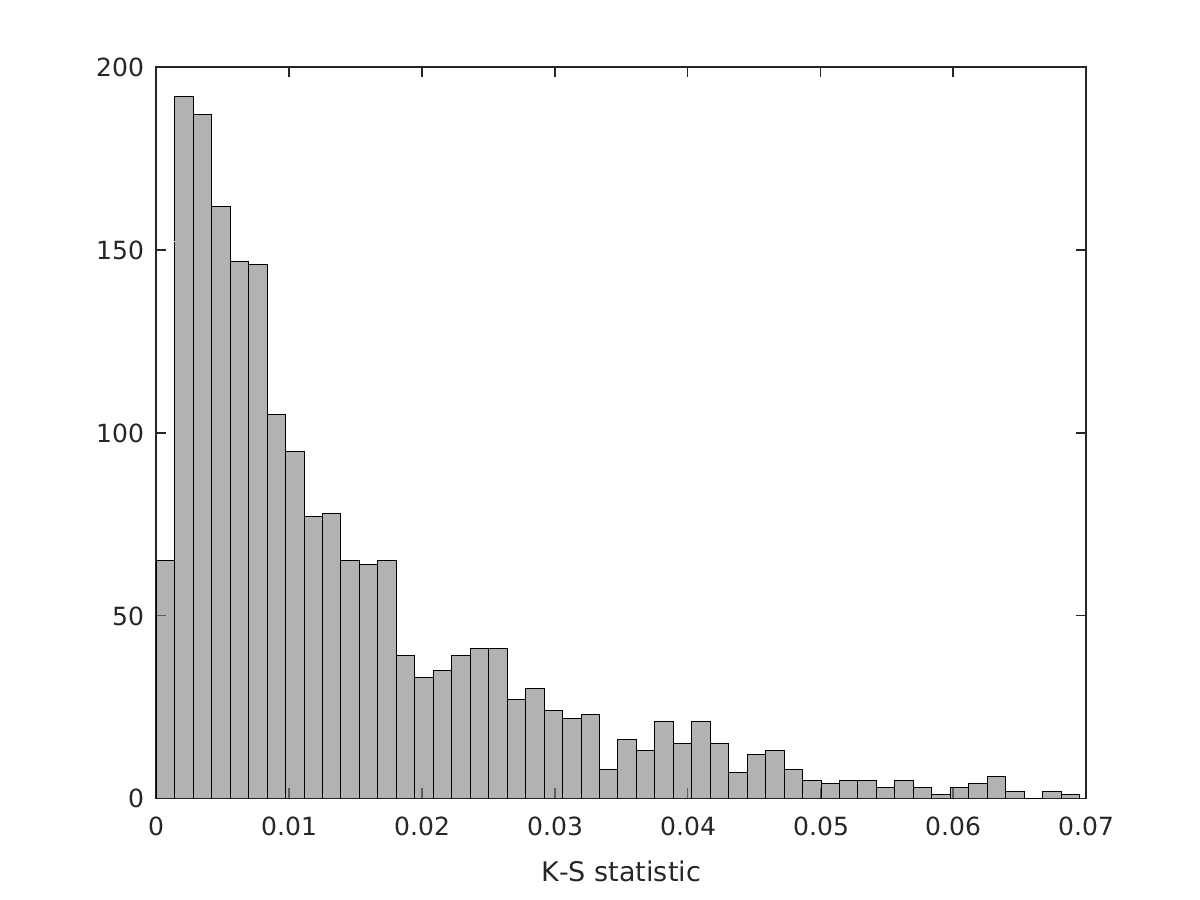} \\
 \end{tabular}
\caption{Distribution of Kolmogorov-Smirnov statistics comparing $\frac{1}{mc} \sum_{t=1}^{mc} \delta_{\theta_t}$ to the oN approximation for 20 randomly selected entries of $\theta$ and over 100 replicate simulations (entries of $\theta$ were re-selected for each replicate).} \label{fig:ks}
\end{figure}

\section{Real Data Example} \label{sec:app}
We consider the Rochdale data, a $2^8$ contingency table with $N=665$ that is over 50 percent sparse, and for which the top ten cell counts all exceed 20. This dataset is described at length in \cite{dobra2011copula}. We first assess the accuracy of the approximation to the full posterior under the Diaconis--Ylvisaker prior in the same manner as in \S \ref{sec:sim}, by comparing marginal posteriors computed using the approximation to those obtained from large Monte Carlo samples from the exact Dirichlet posterior transformed to the log-linear parametrization. For the log-linear model in the corner parametrization, the distribution of Kolmogorov-Smirnov statistics computed for the 255 entries of $\theta^*$ obtained by comparing $10^6$ Monte Carlo samples from the exact posterior to the optimal Gaussian approximation is shown in Fig. \ref{fig:ksroch}. The distribution is very similar to that observed for the simulations in \S \ref{sec:sim}. 

\begin{figure}[h]
\centering
\includegraphics[width=0.4\textwidth]{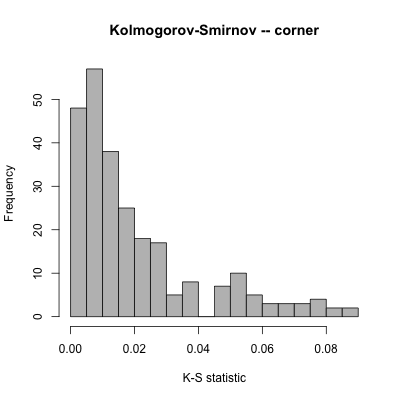}
\caption{Histogram of Kolmogorov-Smirnov statistics for the comparison of $10^6$ Monte Carlo samples from the exact Dirichlet posterior, transformed to $\theta^*$, to the optimal Gaussian approximation to the posterior for $\theta^*$ under the Diaconis--Ylvisaker prior.} \label{fig:ksroch}
\end{figure}

Undoubtedly, the Diaconis--Ylvisaker prior is less well-suited to inference on important variable interactions in this dataset than the more sophisticated methods of \cite{dobra2011copula} and \cite{bhattacharya2012simplex}. However, our approximation has the advantage of being essentially computation-free, whereas the methods of \cite{dobra2011copula} and \cite{bhattacharya2012simplex} are computationally intensive even at this small scale. In many settings, particularly with modern large-scale problems, some loss of performance may be acceptable in order to obtain useful inferences instantaneously. Thus, we are interested in the extent to which our method can replicate the results of \cite{dobra2011copula}, which were similar to those of \cite{bhattacharya2012simplex} in many respects. We analyze performance in testing conditional independence hypotheses (i.e. learning an interaction graph).

Sparse $\theta^*$ is a set of measure zero with respect to the posterior under the Diaconis--Ylvisaker prior. To obtain a sparse point estimate of the interaction graph, we employ the penalized credible region approach of \cite{bondell2012consistent}. This method produces a point estimate by finding the sparsest $\theta^*$ within a $1-\alpha$ credible region for $\theta^*$. Although the exact solution to this problem is intractable, \cite{bondell2012consistent} show that it can be approximated using a lasso path, and provide software in the \texttt{BayesPen R} package  \citep{wilson2015bayes}. Using the resulting lasso path from \texttt{BayesPen}, the selected model corresponding to any value of $\alpha \in (0,1)$ can be obtained as follows. 
\begin{enumerate}[labelindent=*,leftmargin=*]
 \item For the selected value of $\alpha$, find the $1-\alpha$ quantile of a $\chi^2$ distribution with $d-1$ degrees of freedom. Label this $\delta_{\max}$.
 \item For each model $\theta_0$ in the Lasso path, compute the Mahalanobis distance $\delta(\theta_0) = (\theta^*-\theta_0)^T (\Sigma^*)^{-1} (\theta^*-\theta_0)$.
 \item Find the sparsest model in the lasso path having $\delta(\theta_0)\le \delta_{\max}$. This is the sparse point estimate. 
\end{enumerate}
With 256 cells and 665 observations, the posterior under the saturated model with Diaconis--Ylvisaker prior is very diffuse. To make a reasonable comparison, we obtain the posterior under the Diaconis--Ylvisaker prior for the marginal tables corresponding to all ${8 \choose 4} = 70$ unique subsets of four variables. For each of these marginal tables, we then utilize the penalized credible region procedure of \cite{bondell2012consistent} to obtain a sparse model. For comparison, we utilize the median probability graphical model from \cite{dobra2011copula}, which is shown in Table \ref{tab:cggmrochdale}. Specifially, for every subset of four variables, we obtain the marginal graph corresponding to the median probability model of \cite{dobra2011copula} by removing the complement of the subset of nodes under consideration and moralizing, i.e. placing an edge between nodes that (1) have an edge between them in the full graph or (2) are connected solely by a path through nodes that were removed. We treat the graph obtained in this way as the standard for assessing performance of the penalized credible region applied to our Gaussian posterior approximation.

We compute the true (false) negative and positive counts for the penalized credible region procedure applied to our posterior Gausian approximation to all 70 marginal graphs, treating the corresponding marginal median probability graph from \cite{dobra2011copula} as the truth. This produces a total of $70 {4 \choose 2} = 420$ dependent pseudo hypothesis tests. The results for $\alpha = 0.1$ in the penalized credible region procedure are shown in Table \ref{tab:cggmrochdale}. We obtain a false discovery rate of 0.02, and an $F_1$ score of 0.89, indicating that for marginal tables of size $2^4$, the posterior approximation is useful for model selection on the Rochdale data.

\begin{table}
\caption{Left, titled CGGM Results: Marginal posterior inclusion probabilities of edges (above the main diagonal) and indicator of edge inclusion in the median probability model (below the main diagonal) from copula Gaussian graphical model estimated on Rochdale data in \cite{dobra2011copula}. Rows and columns correspond to the eight binary variables, which are labeled a-h. Right, titled Comparison to oN: table of edge classifications for all marginal tables of size $2^4$ from copula Gaussian graphical model median probability model (columns, labeled CGGM) and penalized credible region for Gaussian approximation to posterior under the DY prior (rows, labeled oN-PCR).} \label{tab:cggmrochdale}
\centering
\begin{tabular}[t]{ccc}
CGGM Results & & Comparison to oN \\ 
\begin{tabular}{c|cccccccc}
& a & b & c & d & e & f & g & h \\
\hline
a & -- & 0.93 & 0.67 & 0.92 & 0.32 & 0.42 & 1 & 0.26 \\
b & 1 & -- & 0.27 & 1 & 0.88 & 0.29 & 0.70 & 0.96 \\
c & 1 & 0 & -- & 0.29 & 0.91 & 0.35 & 0.85 & 0.25 \\
d & 1 & 1 & 0 & -- & 0.37 & 0.59 & 0.66 & 0.50 \\
e & 0 & 1 & 1 & 0 & -- & 0.98 & 0.58 & 0.17 \\
f & 0 & 0 & 0 & 1 & 1 & -- & 0.82 & 0.22 \\
g & 1 & 1 & 1 & 1 & 1 & 1 & -- & 0.32 \\
h & 0 & 1 & 0 & 1 & 0 & 0 & 0 & -- \\
\hline
\end{tabular}
&\hspace{10mm} &
\begin{tabular}{cc|cc}
& & \multicolumn{2}{c}{CGGM} \\
& & 0 & 1 \\
\hline
\multirow{2}{*}{oN-PCR} & 0 & 4 & 74 \\
& 1 & 7 & 335 \\
\end{tabular}
\end{tabular}
\\
\end{table}

\section{Discussion}\label{sec:dis}
Outside of linear models, conjugate priors are often non-standard or their multivariate generalizations are difficult to work with. This hampers uncertainty quantification because it is difficult to obtain posterior credible regions for parameters under such priors. Given that automatic and coherent quantification of uncertainty through the posterior is one of the chief advantages of a fully Bayesian approach, this limitation is a significant problem. The optimal Gaussian approximation to the posterior for log-linear models with Dianconis-Ylvisaker conjugate priors derived here offers a highly accurate and essentially computation-free approximation to posterior credible regions for this important class of models. Interestingly, this Gaussian approximation is not the Laplace approximation, and it is faster to compute and offers a better approximation to the posterior than the Laplace approximation. If similar results could be obtained for the posterior in other models, it suggests that the Laplace approximation may not be an appropriate default Gaussian approximation to the posterior. The theoretical result provided here can be easily extended to cases where some categories cannot co-occur, i.e. cases of structural zeros in contingency tables. Extensions to model selection using our approximation are also available by the penalized credible region approach. It seems reasonable that the strategy used here to obtain optimality and convergence rate guarantees could be extended to a larger class of generalized linear models by studying the properties of multivariate Gaussian distributions under inverse link transformations. This may also present a strategy for obtaining approximate credible intervals for parameters in the Bayesian model averaging context for generalized linear models with conjugate priors.

\section*{Acknowledgement}
The authors thank David Dunson for useful conversations and comments during the preparation of this manuscript.

\appendix
\section{Log-linear model details}
The discussion here largely follows \cite{massam2009conjugate} and \cite{lauritzen1996graphical} in its presentation. Let $V$ be the set of variables that will be collected into a contingency table. Let $\mc I_{\gamma}, \gamma \in V$ denote the set of possible levels of values of $\gamma$. Without loss of generality, we can take this set to be a finite collection of sequential nonnegative integers. Let $\mc I = \bigtimes_{\gamma \in V} \mc I_{\gamma}$ be the set of all possible combinations of levels of the variables in $V$. Every cell $i$ of the contingency table corresponds to an element of $V$; thus $|\mc I| = d+1$, where $d$ is defined as in the main text. 

Following \cite{lauritzen1996graphical}, define a cell of the contingency table as $i = (i_{\gamma}, \gamma \in V)$, and let $\pi(i) = \pr[y_1 = i_1,\ldots,y_p = i_p]$. For any $E \subset V$, let $i_E = (i_{\gamma}, \gamma \in E)$ be the cell of the $E$-marginal table corresponding to the values in $i$ of the variables in $E$. Finally, designate the ``base'' cell $i^* = (0,0,\ldots,0)$. Thus, every $i$ can be written as $i = (i_E, i^*_{E^c})$, where $E$ is the subset of $V$ on which $i \ne 0$. Then, the log-linear model in the corner parametrization is given by
\begin{align*}
 \log \frac{\pi(i_E,i^*_{E^c})}{\pi(i^*)} = \sum_{F \subseteq_{\emptyset} E} \theta_F(i_F),
\end{align*}
where for any $F \subset V$, $\theta_F(i_F)$ is a parameter corresponding the the variables in $F$ taking the values in $i_F$, and the notation $\subseteq_{\emptyset}$ means all subsets excluding the empty set. Refer to Proposition 2.1 in \cite{letac2012bayes} for a result showing how the model can be expressed in the form in (\ref{eq:llmX}).

\FloatBarrier
\section{Proof of Proposition \ref{prop:diridy}}
This is readily seen by the change of variable theorem; one only needs some work to calculate the Jacobian term for the change of variable. The matrix of partial derivatives $J = (\partial \theta_j/\partial \pi_r)_{jr}$ is given by 
$$
\frac{\partial \theta_j}{\partial \pi_j} = \frac{1 - \sum_{l \ne j} \pi_l }{\pi_j (1 - \sum_{l=1}^d \pi_l) }, \quad \frac{\partial \theta_j}{\partial \pi_r} = - \frac{1}{1 - \sum_{l=1}^d \pi_l}, \quad (1 \le j \ne r \le d).
$$
Write $J = U + u u^{\T}$, where  $u = (1 - \sum_{l=1}^d \pi_l)^{-1/2} (1, -1, \ldots, -1)^{\T}$ and $U = \mbox{Diag}(1/\pi_1, \ldots, 1/\pi_d)$. We then have $|J| = |U| (1 + u^{\T} U^{-1} u)$ and therefore, 
$$
|J|^{-1} = \pi_1 \ldots \pi_d \bigg(1 - \sum_{l=1}^d \pi_l \bigg) = \frac{e^{\sum_{l=1}^d \theta_l} }{ (1  + \sum_{l=1}^d e^{\theta_l} )^{d+1} }.
$$ 
The proof is concluded by noting that 
$p(\theta; \alpha) = q(\ell(\theta); \alpha) \, |J|^{-1}$. 
\qed

\FloatBarrier

\section{Proof of main results}
We first state some preparatory results that are used to prove the main results. 

\subsection{Preliminaries}
The following identity for the Gamma function is well known (see, e.g., \cite{abramowitz1964handbook}). For $z > 0$, 
\begin{align}\label{eq:gamma_fn_bd}
\Gamma(z) = \frac{\log(2\pi)}{2} + \bigg(z - \frac{1}{2}\bigg) \log z - z + R(z), 
\end{align}
where $0 < R(z) < 1/(12 z)$. 

The digamma function $\psi(z) = \frac{d}{dz} \log \Gamma(z) = \frac{\Gamma'(z)}{\Gamma(z)}$ satisfies $\psi(z+1) = \psi(z) + 1/z$ for any $z > 0$. 
We use the following bound for the digamma function from Lemma 1 of \cite{chen2003best}. For any $z > 0$, 
\begin{align}\label{eq:digamma_bd1}
\frac{1}{2z} - \frac{1}{12z^2} < \psi(z+1) - \log z < \frac{1}{2z}. 
\end{align}
The trigamma function $\psi'(z) = \frac{d}{dz} \psi(z)$ is the derivative of the digamma function. We derive a simple bound for the trigamma function that is used in the sequel.
\begin{lemma}\label{lem:trigamma_bd}
For any $z > 1/3$, 
\begin{align}\label{eq:trigamma_bd_f}
\frac{1}{z} < \psi'(z) < \frac{1}{z} + \frac{1}{z^2}.
\end{align}
The condition $z > 1/3$ is only required for the upper bound. 
\end{lemma}
\begin{proof}
From  \cite{chen2003best}, the trigamma function admits a series expansion
$$
\psi'(z) = \sum_{j=0}^{\infty} \frac{1}{(z+j)^2}
$$
valid for any $z > 0$. The function $t \mapsto t^{-2}$ is monotonically decreasing on $(0, \infty)$ and hence $x^{-2} > \int_{x}^{x+1} t^{-2} dt$ for any $x > 0$. Therefore, for any $z > 0$, $\psi'(z) > \sum_{j=0}^{\infty} \int_{z+j}^{z+j+1} t^{-2} dt = \int_{z}^{\infty} t^{-2} dt = z^{-1}$. For the upper bound, we use Lemma 1 of  \cite{chen2003best} which states that $1/z - \psi'(z+1) > 1/(2z^2) - 1/(6z^3)$ for any $z > 0$. Since $\psi(z+1) = \psi(z) + 1/z$, $\psi'(z+1) = \psi'(z) - 1/z^2$, which yields $\psi'(z) - 1/z < 1/z^2 - 1/(2z^2) + 1/(6z^3) = 1/(2z^2) + 1/(6z^3)$ for any $z > 0$. The conclusion follows since $1/(6z^3) < 1/(2 z^2)$ for any $z > 1/3$.
\end{proof}
Finally, we state a useful result in Lemma \ref{lem:kl}.  
\begin{lemma}\label{lem:kl}
Let $X \in \mb R^d$ be a random vector with $E X = \mu_X$ and $\mbox{var}(X) = \Sigma_X$. For $\mu \in \mb R^d$ and $d \times d$ positive definite matrix $\Sigma$, the mapping 
\begin{align}
(\mu, \Sigma) \mapsto g(\mu, \Sigma) = \log |\Sigma| + E (X - \mu)^{\T} \Sigma^{-1} (X - \mu)
\end{align}
attains its minima when $\mu = \mu_X$ and $\Sigma = \Sigma_X$. The minimum value of the objective function $g(\mu_X, \Sigma_X) = \log |\Sigma_X| + d$. 
\end{lemma}

\begin{proof}
To start with, $E\{ (X - \mu_X)^{\T} \Sigma_X^{-1} (X - \mu_X) \} = \tr [ E \{ (X - \mu_X)(X - \mu_X)^{\T} \Sigma_X^{-1} \}] = \tr(\mr I_d) = d$ and hence $g(\mu_X, \Sigma_X) = \log |\Sigma_X| + d$. Fix $\mu \in \mb R^d$ and $\Sigma$ positive definite. We can write  
\begin{align*}
E\{  (X - \mu)\Sigma^{-1} (X - \mu)  \}
& = \tr[ E \{ (X - \mu) (X - \mu)^{\T} \Sigma^{-1} \} ] \\
& = \tr[ E\{ (X - \mu_X) (X - \mu_X)^{\T} \Sigma^{-1} \} + (\mu_X - \mu)\Sigma^{-1}(\mu_X - \mu) ] \\
& = \tr(\Sigma_X \Sigma^{-1}) + (\mu_X - \mu)^{\T} \Sigma^{-1} (\mu_X - \mu). 
\end{align*}
Therefore, 
\begin{align*}
g(\mu, \Sigma) - g(\mu_X, \Sigma_X) =  \tr(\Sigma_X \Sigma^{-1}) + (\mu_X - \mu)^{\T} \Sigma^{-1} (\mu_X - \mu) - d - \log |\Sigma_X \Sigma^{-1}|.
\end{align*}
The above quantity is non-negative since it equals $2 D\big\{N(\mu_X, \Sigma_X) \dmid N(\mu, \Sigma)\big\}$, i.e., twice the Kullback--Leibler divergence between $N(\mu_X, \Sigma_X)$ and $N(\mu, \Sigma)$. Since $\mu$ and $\Sigma$ were arbitrary, the first part is proved. The second part has been already proved at the beginning. 
\end{proof}

\subsection{Proof of Theorem \ref{thm:KL} and Corollary \ref{cor:KL}}
We can now give a proof of Theorem \ref{thm:KL}. Recall the Dirichlet density $q$ from \eqref{eq:dir_dens} and the logistic normal density $\wt{q}$ from \eqref{eq:ln_dens}. We shall write $q(\pi)$ and $\wt{q}(\pi)$ in place of $q(\pi \mid \beta)$ and $\wt{q}(\pi \mid \mu, \Sigma)$ henceforth for brevity. From \eqref{eq:dir_dens} and \eqref{eq:ln_dens}, 
\begin{align*}
\log \frac{q(\pi)}{\wt{q}(\pi)} = \log B_{\beta} + \frac{d \log (2 \pi)}{2}  + \sum_{j=0}^d \beta_j \log \pi_j + + \frac{\log |\Sigma|}{2}  + \frac{1}{2} \big\{ \log(\pi/\pi_0) - \mu \big\}^{\T} \Sigma^{-1} \big\{ \log(\pi/\pi_0) - \mu \big\}.
\end{align*}
Observe that $\mu$ and $\Sigma$ appear only in the last two terms in the right hand side of the above display. Invoking Lemma \ref{lem:kl}, it is therefore evident that $D(q \dmid \wt{q}) = E_q \log(q/\wt{q})$ is minimized when $\mu^* = E_q \log(\pi/\pi_0)$ and $\Sigma^* = \mbox{var}_q \{\log(\pi/\pi_0)\}$, and the minimum vaue of the Kullback--Leibler divergence is
\begin{align}\label{eq:min_kl_1}
\log B_{\beta} + \sum_{j=0}^d \beta_j E_q \log \pi_j + \frac{d}{2}\{1 + \log (2 \pi)\} + \frac{\log |\Sigma^*|}{2}.
\end{align} 

Using standard  properties of the Dirichlet distribution or Exponential family differential identities, with $\beta = \sum_{j=0}^d \beta_j$, 
\begin{align}
& E_q \log \pi_j = \psi(\beta_j) - \psi(\beta), \quad j = 0, 1, \ldots d, \label{eq:dir_prop1}\\
& \mbox{cov}_q(\log \pi_j, \log \pi_l) = \psi'(\beta_j) \delta_{jl} - \psi'(\beta), \quad j, l = 0, 1, \ldots, d. \label{eq:dir_prop2}
\end{align}
Therefore, $\mu_j^* = E_q \log \pi_j - E_q \log \pi_0 = \psi(\beta_j) - \psi(\beta_0)$ for $j=1, \ldots d$. Next, $\sigma_{jj'}^* = \mbox{cov}_q(\log \pi_j - \log \pi_0, \log \pi_{j'} - \log \pi_0) = \delta_{jj'} \psi'(\beta_j) + \psi'(\beta_0) $ for $j, j' = 1, \ldots, d$. The expressions for $\mu^*$ and $\Sigma^*$ are identical to \eqref{eq:min_mus}, proving the first part of the theorem. Note this also establishes Proposition \ref{prop:meancov}.

We now proceed to bound each term in the expression for the minimum Kullback--Leibler divergence in \eqref{eq:min_kl_1}; refer to them by $T_1, T_2, T_3$ and $T_4$ respectively. First, we have,
\begin{align}
T_1 & := \log B_{\beta} = \log \Gamma(\beta) - \sum_{j=0}^d \Gamma(\beta_j) \notag \\
& < - \frac{d \log (2\pi)}{2} + \bigg( \beta \log \beta - \sum_{j=0}^d \beta_j \log \beta_j \bigg) - \frac{1}{2} \bigg(\log \beta - \sum_{j=0}^d \log \beta_j \bigg) + \frac{1}{12 \beta}. \label{eq:bd_T1}
\end{align}
In the above display, we used \eqref{eq:gamma_fn_bd} to bound $\log \Gamma(\beta)$ from above and $\log \Gamma(\beta_j)$s from below. 
The $(-\beta)$ term in upper bound to $\log \Gamma(\beta)$ cancels out the $(-\sum_{j=0}^d \beta_j)$ contribution from the lower bounds to the $\log \Gamma(\beta_j)$s. Next, 
\begin{align}
T_2 
&:= \sum_{j=0}^d \beta_j E_q \pi_j = \sum_{j=0}^d \beta_j \{\psi(\beta_j) - \psi(\beta)\} \notag \\
& = \sum_{j=0}^d \beta_j \{ \psi(\beta_{j+1}) - \psi(\beta+1) \} - \sum_{j=0}^d \beta_j \bigg(\frac{1}{\beta_j} - \frac{1}{\beta} \bigg) \notag \\
& = \bigg\{\sum_{j=0}^d \beta_j \psi(\beta_{j+1}) - \beta \psi(\beta) \bigg\} - d \notag \\
& < \bigg(\sum_{j=0}^d \beta_j \log \beta_j - \beta \log \beta\bigg) - \frac{d}{2} + \frac{1}{12 \beta}. \label{eq:bd_T2}
\end{align}
In the first line of the above display, we used \eqref{eq:dir_prop1}. From the first to the second line, we used the identity $\psi(z+1) = \psi(z) + 1/z$. From the second to the third line, we only use $\sum_{j=0}^d \beta_j = \beta$. From the third to the fourth line, we made use of the bound \eqref{eq:digamma_bd1} for the digamma function $\psi$. From the upper bound in \eqref{eq:digamma_bd1}, $\beta_j \psi(\beta_{j+1}) < \beta_j \log \beta_j + 1/2$ and hence $\sum_{j=0}^d \beta_j \psi(\beta_{j+1}) < \sum_{j=0}^d \beta_j \log \beta_j + (d+1)/2$. From the lower bound in \eqref{eq:digamma_bd1}, $\beta \psi(\beta) > \beta \log \beta + 1/2 - 1/(12 \beta)$. 

Finally, from \eqref{eq:dir_prop2}, we can write $\Sigma^* = D + \psi'(\beta_0) \bs 1 \bs{1}^{\T}$, with $D = \mbox{diag}(\psi'(\beta_1), \ldots, \psi'(\beta_d))$. Using the fact $|X + u v^{\T}| = |X| (1 + v^{\T} X^{-1} u)$, we obtain
\begin{align*}
|\Sigma^*| = \bigg\{ 1 + \sum_{j=1}^d \psi'(\beta_0)/\psi'(\beta_j) \bigg\} \bigg\{ \prod_{j=1}^d \psi'(\beta_j) \bigg\} 
= \bigg\{ \sum_{j=0}^d \frac{\psi'(\beta_0)}{\psi'(\beta_j)} \bigg\} \bigg\{ \prod_{j=1}^d \psi'(\beta_j) \bigg\}.
\end{align*}
From Lemma \ref{lem:trigamma_bd}, $\psi'(\beta_j) > 1/\beta_j$, implying
\begin{align}
T_4 
&:= \frac{\log |\Sigma^*|}{2} = \frac{1}{2} \bigg[ \log \bigg\{ \sum_{j=0}^d \frac{\psi'(\beta_0)}{\psi'(\beta_j)} \bigg\} + \sum_{j=1}^d \log \psi'(\beta_j) \bigg] \notag \\
& < \frac{1}{2} \bigg\{ \log \beta + \sum_{j=0}^d \log \psi'(\beta_j) \bigg\}. \label{eq:bd_T4}
\end{align} 
Recalling $T_3 = d \{1 + \log (2 \pi)\}/2$ and substituting the bounds for $T_1, T_2$ and $T_4$ from \eqref{eq:bd_T1}, \eqref{eq:bd_T2} and \eqref{eq:bd_T4} in \eqref{eq:min_kl_1}, we obtain, after plenty of cancellations,
\begin{align*}
\sum_{j=1}^4 T_j 
&< \frac{1}{2} \sum_{j=0}^d \log \{ \beta_j \psi'(\beta_j) \} + \frac{1}{6 \beta}  \\
& < \frac{1}{2} \sum_{j=0}^d \frac{1}{\beta_j} + \frac{1}{6 \beta}. 
\end{align*}
From the first to the second line, we invokeed Lemma \ref{lem:trigamma_bd} to bound $\beta_j \psi'(\beta_j) < 1 + 1/\beta_j$ and used $\log(1+x) < x$ for $x > 0$. We have obtained the desired bound, concluding the proof. 

Now, to show Corollary \ref{cor:KL}, just note that 
 by the invariance of $D$ under one-to-one transformations, we have that for any full rank matrix $X$,
\begin{align}
D \bigg\{ \mc{LD}(\beta) \dmid \mc{N}(\mu, \Sigma) \bigg\} = D \bigg\{ \mc{P}_X(\cdot;\beta) \dmid \mc{N}(X \mu, X^{\T} \Sigma X) \bigg\}. \label{eq:kltrans}
\end{align}
So 
\begin{align}
\inf_{\mu, \Sigma} \bigg\{ \mc{LD}(\beta) \dmid \mc{N}(\mu, \Sigma) \bigg\} = \inf_{\widetilde{\mu}, \widetilde{\Sigma}} D \bigg\{ \mc{P}_X(\cdot;\beta) \dmid \mc{N}(\widetilde{\mu}, \widetilde{\Sigma}) \bigg\}. \label{eq:klinf}
\end{align}
Since the infimum on the left side in (\ref{eq:klinf}) is attained by $\mu^*, \Sigma^*$, we have by (\ref{eq:kltrans}) that
$$  D\left( \mc P_X(\cdot;\beta) \dmid \mc N(\cdot;X \mu^*,X^T \Sigma^* X) \right) = \inf_{\mu, \Sigma} D\left( \mc P_X(\cdot;\beta) \dmid \mc N(\cdot;\mu,\Sigma) \right), $$
which gives Corollary \ref{cor:KL}.
\bibliographystyle{plainnat}
\bibliography{dynormal_arXiv}

\begin{thebibliography}{26}
\expandafter\ifx\csname natexlab\endcsname\relax\def\natexlab#1{#1}\fi

\bibitem[{Abramowitz \& Stegun(1964)}]{abramowitz1964handbook}
\textsc{Abramowitz, M.} \& \textsc{Stegun, I.~A.} (1964).
\newblock \textit{Handbook of mathematical functions: with formulas, graphs,
  and mathematical tables}.
\newblock No.~55. Courier Corporation.

\bibitem[{Agresti(2002)}]{agresti2002categorical}
\textsc{Agresti, A.} (2002).
\newblock \textit{Categorical data analysis}, vol. 359.
\newblock John Wiley \& Sons.

\bibitem[{Attias(1999)}]{attias1999inferring}
\textsc{Attias, H.} (1999).
\newblock Inferring parameters and structure of latent variable models by
  variational bayes.
\newblock In \textit{Proceedings of the Fifteenth conference on Uncertainty in
  artificial intelligence}. Morgan Kaufmann Publishers Inc.

\bibitem[{Bhattacharya \& Dunson(2012)}]{bhattacharya2012simplex}
\textsc{Bhattacharya, A.} \& \textsc{Dunson, D.~B.} (2012).
\newblock Simplex factor models for multivariate unordered categorical data.
\newblock \textit{Journal of the American Statistical Association}
  \textbf{107}, 362--377.

\bibitem[{Bishop et~al.(2007)Bishop, Fienberg \& Holland}]{bishop2007discrete}
\textsc{Bishop, Y.~M.}, \textsc{Fienberg, S.~E.} \& \textsc{Holland, P.~W.}
  (2007).
\newblock \textit{Discrete multivariate analysis: theory and practice}.
\newblock Springer Science \& Business Media.

\bibitem[{Bondell \& Reich(2012)}]{bondell2012consistent}
\textsc{Bondell, H.~D.} \& \textsc{Reich, B.~J.} (2012).
\newblock Consistent high-dimensional bayesian variable selection via penalized
  credible regions.
\newblock \textit{Journal of the American Statistical Association}
  \textbf{107}, 1610--1624.

\bibitem[{Chen \& Qi(2003)}]{chen2003best}
\textsc{Chen, C.-P.} \& \textsc{Qi, F.} (2003).
\newblock The best lower and upper bounds of harmonic sequence.
\newblock \textit{RGMIA research report collection} \textbf{6}.

\bibitem[{Dellaportas \& Forster(1999)}]{dellaportas1999markov}
\textsc{Dellaportas, P.} \& \textsc{Forster, J.~J.} (1999).
\newblock Markov chain monte carlo model determination for hierarchical and
  graphical log-linear models.
\newblock \textit{Biometrika} \textbf{86}, 615--633.

\bibitem[{Diaconis \& Ylvisaker(1979)}]{diaconis1979conjugate}
\textsc{Diaconis, P.} \& \textsc{Ylvisaker, D.} (1979).
\newblock Conjugate priors for exponential families.
\newblock \textit{The Annals of statistics} \textbf{7}, 269--281.

\bibitem[{Dobra \& Lenkoski(2011)}]{dobra2011copula}
\textsc{Dobra, A.} \& \textsc{Lenkoski, A.} (2011).
\newblock Copula gaussian graphical models and their application to modeling
  functional disability data.
\newblock \textit{The Annals of Applied Statistics} \textbf{5}, 969--993.

\bibitem[{Dobra \& Massam(2010)}]{dobra2010mode}
\textsc{Dobra, A.} \& \textsc{Massam, H.} (2010).
\newblock The mode oriented stochastic search (moss) algorithm for log-linear
  models with conjugate priors.
\newblock \textit{Statistical Methodology} \textbf{7}, 240--253.

\bibitem[{Fienberg \& Rinaldo(2007)}]{fienberg2007three}
\textsc{Fienberg, S.~E.} \& \textsc{Rinaldo, A.} (2007).
\newblock Three centuries of categorical data analysis: Log-linear models and
  maximum likelihood estimation.
\newblock \textit{Journal of Statistical Planning and Inference} \textbf{137},
  3430--3445.

\bibitem[{Gelfand \& Smith(1990)}]{gelfand1990sampling}
\textsc{Gelfand, A.~E.} \& \textsc{Smith, A.~F.} (1990).
\newblock Sampling-based approaches to calculating marginal densities.
\newblock \textit{Journal of the American statistical association} \textbf{85},
  398--409.

\bibitem[{Haberman(1974)}]{haberman1974log}
\textsc{Haberman, S.~J.} (1974).
\newblock Log-linear models for frequency tables derived by indirect
  observation: Maximum likelihood equations.
\newblock \textit{The Annals of Statistics} , 911--924.

\bibitem[{Hoeting et~al.(1998)Hoeting, Madigan, Raftery \&
  Volinsky}]{hoeting1998bayesian}
\textsc{Hoeting, J.~A.}, \textsc{Madigan, D.}, \textsc{Raftery, A.~E.} \&
  \textsc{Volinsky, C.~T.} (1998).
\newblock Bayesian model averaging.
\newblock In \textit{In Proceedings of the AAAI Workshop on Integrating
  Multiple Learned Models}. Citeseer.

\bibitem[{Lauritzen(1996)}]{lauritzen1996graphical}
\textsc{Lauritzen, S.~L.} (1996).
\newblock \textit{Graphical models}.
\newblock Oxford University Press.

\bibitem[{Letac \& Massam(2012)}]{letac2012bayes}
\textsc{Letac, G.} \& \textsc{Massam, H.} (2012).
\newblock Bayes factors and the geometry of discrete hierarchical loglinear
  models.
\newblock \textit{The Annals of Statistics} \textbf{40}, 861--890.

\bibitem[{Massam et~al.(2009)Massam, Liu \& Dobra}]{massam2009conjugate}
\textsc{Massam, H.}, \textsc{Liu, J.} \& \textsc{Dobra, A.} (2009).
\newblock A conjugate prior for discrete hierarchical log-linear models.
\newblock \textit{The Annals of Statistics} \textbf{37}, 3431--3467.

\bibitem[{Park \& Hastie(2007)}]{park2007l1}
\textsc{Park, M.~Y.} \& \textsc{Hastie, T.} (2007).
\newblock L1-regularization path algorithm for generalized linear models.
\newblock \textit{Journal of the Royal Statistical Society: Series B
  (Statistical Methodology)} \textbf{69}, 659--677.

\bibitem[{Polson et~al.(2013)Polson, Scott \& Windle}]{polson2013bayesian}
\textsc{Polson, N.~G.}, \textsc{Scott, J.~G.} \& \textsc{Windle, J.} (2013).
\newblock Bayesian inference for logistic models using p{\'o}lya--gamma latent
  variables.
\newblock \textit{Journal of the American Statistical Association}
  \textbf{108}, 1339--1349.

\bibitem[{Shun \& McCullagh(1995)}]{shun1995laplace}
\textsc{Shun, Z.} \& \textsc{McCullagh, P.} (1995).
\newblock Laplace approximation of high dimensional integrals.
\newblock \textit{Journal of the Royal Statistical Society. Series B
  (Methodological)} , 749--760.

\bibitem[{Tierney \& Kadane(1986)}]{tierney1986accurate}
\textsc{Tierney, L.} \& \textsc{Kadane, J.~B.} (1986).
\newblock Accurate approximations for posterior moments and marginal densities.
\newblock \textit{Journal of the american statistical association} \textbf{81},
  82--86.

\bibitem[{Wang \& Titterington(2004)}]{wang2004lack}
\textsc{Wang, B.} \& \textsc{Titterington, D.} (2004).
\newblock Lack of consistency of mean field and variational bayes
  approximations for state space models.
\newblock \textit{Neural Processing Letters} \textbf{20}, 151--170.

\bibitem[{Wang \& Titterington(2005)}]{wang2005inadequacy}
\textsc{Wang, B.} \& \textsc{Titterington, D.} (2005).
\newblock Inadequacy of interval estimates corresponding to variational
  bayesian approximations.
\newblock \textit{Proc. 10th Int. Wrkshp Artificial Intelligence and
  Statistics} , 373--380.

\bibitem[{Wilson et~al.(2015)Wilson, Bondell \& Reich}]{wilson2015bayes}
\textsc{Wilson, A.}, \textsc{Bondell, H.~D.} \& \textsc{Reich, B.~J.} (2015).
\newblock Bayespen: Bayesian penalized credible regions. r package version 1.2.

\bibitem[{Zou \& Hastie(2005)}]{zou2005regularization}
\textsc{Zou, H.} \& \textsc{Hastie, T.} (2005).
\newblock Regularization and variable selection via the elastic net.
\newblock \textit{Journal of the Royal Statistical Society: Series B
  (Statistical Methodology)} \textbf{67}, 301--320.

\end{thebibliography}


%
%
%

\end{document}